\documentclass{sig-alternate}
\usepackage{amsmath,amsopn,amscd,amssymb,xypic,rotating,epic}
\newcommand{\double}       {}

\xyoption{all}

\newtheorem{theorem}{Theorem}[section]
\newtheorem{lemma}{Lemma}[section]

\newcommand{\cb}          {\begin{tabbing}MMMMM\=MM\=MM\=MM\=MM\=MM\=MM\=MM\=MM\=MM\= \kill}
\newcommand{\ce}          {\end{tabbing}}

\def\eor{\\\hfill(end of remark)}

\newcommand{\separate}     {\vspace{0.3cm}\begin{center}*~~~~~~~~~~*~~~~~~~~~~*\end{center}\vspace{0.3cm}}
\newcommand{\dqt}[1]        {"{#1}"}

\def\emph{\textsl}
\def\em{\sl}
\def\dfeq{\stackrel{\triangle}{=}}

\newcommand{\bra}[1]     {\langle{#1}|}
\newcommand{\ket}[1]     {|{#1}\rangle}

\newcommand{\bracket}[2] {\langle{#1}|{#2}\rangle}

\newcommand{\mtx}[1]     {\mathbf{{#1}}}

\newcounter{myremark}
\newcounter{myexample}
\def\remark{
\refstepcounter{myremark}%
\noindent \emph{Remark \arabic{myremark}:}
}

\begin{document}

\title{Network navigation using Page Rank random walks}
 \numberofauthors{2}
 \author{    
   \alignauthor Emilio Aced Fuentes\\
   \affaddr{Escuela Polit\'ecnica Superior}\\
   \affaddr{Universidad Aut\'onoma de Madrid}\\
   \email{emilio.aced@estudiante.uam.es}
   \alignauthor Simone Santini\titlenote{Simone Santini was supported in part by the 
       Spanish \emph{Ministerio de Ciencia e Innovaci\'on} under the grant N. PID2019-108965GB-I00 \emph{M\'as all\'a
       de la recomendaci\'on est\'atica: equidad, interacci\'on y transparencia}}\\
   \affaddr{Escuela Polit\'ecnica Superior}\\
   \affaddr{Universidad Aut\'onoma de Madrid}\\
   \email{simone.santini@uam.es}
 }

\maketitle

\double

\begin{abstract}
    We introduce a formalism based on a continuous time approximation, to study the characteristics 
    of Page Rank random walks. We find that the diffusion of the occupancy probability has a dynamics
    that exponentially \dqt{forgets} the initial conditions and settles to a steady state that depends
    only on the characteristics of the network. In the special case in which the walk begins from a single 
    node, we find that the largest eigenvalue of the transition value ($\lambda_1=1$) does not contribute 
    to the dynamic and that the probability is constant in the direction of the corresponding eigenvector.
    We study the process of visiting new node, which we find to have a dynamic similar to that of the
    occupancy probability. Finally, we determine the average transit time between nodes $\langle{T}\rangle$,
    which we find to exhibit certain connection with the corresponding time for L\'evy walks.
    The relevance of these results reside in that Page Rank, which are a more reasonable model for
    the searching behavior of individuals, can be shown to exhibit features similar to L\'evy walks,
    which i nturns have been shown to be a reasonable model of a common large scale search strategy
    known as \emph{Area Restricted Search}.
\end{abstract}

\section{Introduction}
The study of stochastic processes in graphs has a fairly long history
in the analysis of transport and diffusion processes, from the study
of epidemics \cite{barbour:90,britton:07} and technical systems
\cite{vespignani:12} to animal \cite{ramos:04,boyer:06} and human
\cite{brockmann:06,gonzalez:08} movements or the analysis of social
interactions \cite{iribarren:09}. With the rise in popularity of the
Internet around the turn of the century and of social networks a
decade later, there has been a corresponding surge of interest in
the study of random walks on graphs, this time seeing them as models
of search strategies: a possible and common search strategy is, in
fact, jumping from neighbor to neighbor until the information is found
that one is after. In a standard random walks on a
graph, one moves at each time step from a node to one of its neighbors
picked at random  \cite{lovasz:96}. These walks are a classic whose study precedes the
applications to search and their characteristics are well
understood. However, the connection between random walks and search
has generated new interest in other kinds of walks.

In the continuum, for example, it has been observed that the so-called
\emph{L\'evy Walks} (walks in which one makes jumps at a distance $d$
from the current location with probability $p(d)\sim{d^{-\alpha}}$
\cite{ali:18,borgers:20}) exhibit diffusion characteristics similar to
those of a common animal behavior known as \emph{Area Restricted
  Search (ARS)} \cite{riascos:14,santini:20a}. The clearest example of ARS is foraging: an
animal will move around with small movements, staying essentially in
the same patch as long as there is a lot of food to be found then, as
the food becomes scarce, will do a long migration to find a new
patch. This behavior, which in animals is controlled by dopamine, is
at the base of many problem-solving behaviors in virtually all
animals%
\footnote{More precisely: the behavior is found in all
  \emph{eumetazoa} that is, in all animals except \emph{porifera}, a
  class that contains sponges and little more.}%
, from the nematode \emph{C.elegans} hunting for food, to the saccadic
eye movements of a person looking at an image. The reason for
this universality (apart from the early appearance of dopamine control
in evolution) is that ARS is optimal for \dqt{patchy} resources in
which the location and characteristics of the patches are not known in
advance \cite{santini:20a}.

In the continuum, L\'evy walks have a dynamic behavior that mimics
that of ARS walks, and this has led to their study in connection with
search problems \cite{santini:20a}. The analysis has been extended to
L\'evy walks on graphs, in which the walker jumps from a node to
another at a distance $d$ with probability $p(d)\sim{d^{-\alpha}}$, as
in the continuous case, but where now $d$ is the length of the
shortest path between two nodes \cite{riascos:12}. This study has led
to useful insight into the behavior of these walks on graphs.

One issue that one might rise in connection with this work is whether
a L\'evy walk is a suitable model of search behavior. When one is
searching the web, for example, an with the exception of the immediate
neighbors ($d=1$), one has no idea what pages are at a shortest-path
distance $d$ from the one they are looking at. In a social network, one
might conceivably jump directly at the site of a friend of a friend
($d=2$) but, beyond that, it is virtually impossible to decide to
execute a jump at a distance $d$.

One, more reasonable, model seems to be the following. A person keeps
looking for information moving from one page to another but at any
given time, with probability $u$, she invokes a second mechanism
(e.g.\ a search engine) that leads to a new page that has no
link-structural relation with the one she was just visiting. That is,
from the point of view of a random walk, the behavior is the
following: at any moment, with probability $1-u$ move to one of your
neighbors chosen at random; with probability $u$, jump to a random
node in the graph. This type of walk is known as the \emph{Page Rank}
random walk%
\footnote{The name derives from the algorithm developed in \cite{page:99} to
  assign \dqt{structural relevance} to web sites and originally used
  in the \emph{google} search engine: the relevance of a site is the
  probability of being in that site in the stationary solution of the
  Page Rank random walk.}%
~\cite{brin:98,page:99}. 

In this paper, we analyze the Page Rank random walk and derive some
of its features, most notably, the dynamics that leads to the
stationary solution, the number of new nodes visited as a function of
time, and the average time of the walk between two random nodes of the
graph.

\section{Basic Equations}
Let $G=(V,E)$ be an undirected graph with $n$ nodes and $m$ edges
($V=\{1,\ldots,n\}$, $E\subseteq{V}\times{V}$, $|E|=2m$), with
adjacency matrix $\mtx{A}$ ($a_{ij}=1$ if $(i,j)\in{E}$, $a_{ij}=0$
otherwise), where $\mtx{A}=\mtx{A}'$ and $a_{ii}=0$ (we do not allow
self-loops). Let $d_i=\sum_ja_{ij}$ be the degree of node $i$, that
is, the number of its neighbors. Also, define the matrix
$\mtx{D}=\mbox{diag}(d_1,\ldots,d_n)$.

In a random walk, we start at time $t=0$ from a node $v_0\in{V}$ and,
if at time $t$ we are in node $v_t$, at time $t+1$ we move to one of
its neighbors with probability $1/d_{v_t}$. Define the matrix
$\mtx{W}=\mtx{D}^{-1}\mtx{A}$ with elements
\begin{equation}
  w_{ij} = \frac{a_{ij}}{d_i} = 
    \begin{cases}
      \frac{1}{d_i} & \mbox{if $(i,j)\in{E}$} \\
      0             & \mbox{otherwise}
    \end{cases}
\end{equation}
Then $w_{ij}$ is the probability of moving from node $i$ to node $j$
in a single step. Let $p_i(t)={\mathbb{P}}[v_t=i]$ be the probability
of being in node $i$ at time $t$. The evolution of this probability is
given by the \emph{Master Equation}
\begin{equation}
  p_i(t+1) = \sum_{j:(i,j)\in{E}} \frac{1}{d_j} p_j(t) = 
  \sum_{j=1}^n \frac{a_{ij}}{d_j} p_j(t) = 
  \sum_{j=1}^n w_{ji} p_j(t) 
\end{equation}
If we collect the probabilities in a vector
$\ket{p}(t)=[p_1(t),\ldots,p_n(t)]'$ then the master equation can be
rewritten as
\begin{equation}
  \ket{p}(t+1) = \mtx{W}'\ket{p}(t)
\end{equation}
That is, if $\ket{p}^0$ is the initial probability distribution of the
walker
\begin{equation}
  \ket{p}(t) = (\mtx{W}')^t \ket{p^0}
\end{equation}
It is easy to check that the probability distribution $\ket{\pi}$, with
\begin{equation}
  \pi_i = \frac{d_i}{\sum_k d_k} = \frac{d_1}{2m}
\end{equation}
satisfies 
\begin{equation}
  \label{pieigen}
  \ket{\pi}=\mtx{W}'\ket{\pi}
\end{equation}
and is therefore a stationary distribution of the walk. It is also
possible to show that this distribution is unique \cite{lovasz:96}. From
(\ref{pieigen}) it follows that $\ket{\pi}$ is an eigenvector of
$\mtx{W}'$ corresponding to $\lambda_1=1$ (or, equivalently, a right
eigenvector of $\mtx{W}$: $\bra{\pi}=\bra{\pi}\mtx{W}$) and the
unicity of the stationary distribution implies that the eigenvalue
$\lambda_1=1$ has multiplicity $1$. Since the eigenvector $\ket{\pi}$
has all positive components, it follows from the Frobenius-Perron
theorem \cite{biyikoglu:07} that
\begin{equation}
  1 = \lambda_1 > \lambda_2 \ge \lambda_3 \ge \cdots \ge \lambda_n \ge -1
\end{equation}

We show in Appendix \ref{bipartite} that $\lambda_n=-1$ only for
bipartite graphs, a case that we do not consider, so we shall always
assume $\lambda_n>-1$.  If we are at a node $i$, the probability of
moving to one of its neighbor is one. Consequently
\begin{equation}
  \label{coo}
  \sum_j w_{ij} = \sum_j \frac{a_{ij}}{d_i} = \frac{1}{d_i} \sum_{(i,j)\in{E}} 1 = 1 
\end{equation}
If $\ket{1_n}=[1,\ldots,1]'$, it is easy to see from (\ref{coo}) that
$\mtx{W}\ket{1_n}=\ket{1_n}$, or
\begin{equation}
  \label{pip}
  \bra{1_n} = \bra{1_n}\mtx{W}'
\end{equation}
That is, the constant vector $\ket{1_n}$ is the right
eigenvector of $\mtx{W}'$ corresponding to $\lambda_1=1$.

The left eigenvectors other than the first correspond to eigenvalues
with $|\lambda|<1$. These have a property that will be quite relevant
in our context:

\begin{lemma}
  \label{sum0}
  Let $\ket{v}$ be a left eigenvector of $\mtx{W}'$ corresponding to an
  eigenvalue $\lambda\ne{1}$. Then
  \begin{equation}
    \label{bingo}
    \sum_i v_i = 0
  \end{equation}

\end{lemma}

\begin{proof}
  The eigenvector equation gives
  \begin{equation}
    \sum_j w_{ji}v_j = \lambda v_i
  \end{equation}
  Summing the components we have
  \begin{equation}
    \lambda \sum_i v_i = \sum_i \sum_j w_{ji}v_j = \sum_j \underbrace{\sum_i w_{ji}}_{1} v_j = \sum_j v_j
  \end{equation}
  where the sum over $j$ in the third expression is equal to 1
  because of (\ref{coo}). The equation
  \begin{equation}
    \lambda \sum_i v_i = \sum_j v_j
  \end{equation}
  with $\lambda\ne{1}$ has (\ref{bingo}) as the only solution
\end{proof}
%
%

%
%  NOT SURE ABOUT THE PLACEMENT OF THIS
%
The
matrix $\mtx{W}'$ can be diagonalized as
\begin{equation}
  \mtx{W}' = \mtx{T}\mtx{\Lambda}\mtx{T}^{-1}
\end{equation}
Let $\mtx{T}$ and $\mtx{T}^{-1}$ be defined as 
\begin{equation}
  \label{tdef}
  \mtx{T} = \bigl[\ket{\phi_1} | \cdots | \ket{\phi_n} \bigr]
\end{equation}
and 
\begin{equation}
  \label{t1def}
  \mtx{T}^{-1} = \left[
    \begin{array}{c}
      \bra{\psi_1} \\
      -- \\
      \vdots \\
      -- \\
      \bra{\psi_n}
    \end{array}
    \right]
\end{equation}
Then $\ket{\phi_i}$ are eigenvectors and $\bra{\psi_i}$ 
right eigenvectors of $\mtx{W}'$, that is
\begin{equation}
  \begin{aligned}
    \mtx{W}'\ket{\phi_i} &= \lambda_i\ket{\phi_i} \\
    \bra{\psi_i} \mtx{W}' &= \lambda_i \bra{\psi_i}
  \end{aligned}
\end{equation}

Let $t_{ij}$ be the $i,j$ element of $\mtx{T}$ and $\tau_{ij}$ that of
$\mtx{T}^{-1}$, then
$t_{ij}=\bracket{e_i}{\phi_j}=\phi_{j,i}$ (the $i$th component of $\ket{\phi_j}$),
and
$\tau_{ij}=\bracket{\psi_i}{e_j}=\psi_{i,j}$ (the $j$th component of $\bra{\psi_i}$).
We have
\begin{equation}
  \label{forgot}
  \delta_{ij} = \bigl(\mtx{T}^{-1}\mtx{T})_{ij} = 
  \sum_k \tau_{ik}t_{kj} = \sum_k \psi_{i,k} \phi_{j,k} = \bracket{\psi_i}{\phi_j}
\end{equation}
and
\begin{equation}
  \label{deltaeq}
  \delta_{ij} = \bigl(\mtx{T}\mtx{T}^{-1})_{ij} = 
  \sum_k t_{ik}\tau_{kj} = \sum_k \phi_{i,k} \psi_{j,k} 
\end{equation}

The matrix $\mtx{W}'$ is that of the standard random walk, and its
only equilibrium point is the probability vector $\ket{\pi}$. This
vector satisfies the equation $\mtx{W}'\ket{\pi}=\ket{\pi}$, that is,
$\pi$ is an eigenvector relative to $\lambda=1$ and therefore
$\phi_1\sim\pi$, or
\begin{equation}
  \phi_{1,i} = b d_i
\end{equation}
with $b>0$. 
From (\ref{pip}) we have
\begin{equation}
  \bra{1_n} \mtx{W}' = \bra{1_n}
\end{equation}
that is, 
\begin{equation}
  \bra{\psi_1} = a \bra{1_n}
\end{equation}
with $a>0$. One of the constants $a$ and $b$ can be determined by
condition (\ref{forgot}) (the other one is arbitrary):
\begin{equation}
  1 = \bracket{\psi_1}{\phi} = a\cdot b\cdot \sum_i d_i
\end{equation}
We can choose $b=1/\sum{d_i}$, leading to
\begin{equation}
  \phi_{1,i} = \frac{d_i}{\sum_k d_k} = \pi_i
\end{equation}
and $a=1$, that is $\bra{\psi_1}=\bra{1_n}$. Note that
\begin{equation}
  \label{sums}
  \begin{aligned}
    \sum_k \phi_{1,k} &= 1 \\
    \sum_k \psi_{1,k} &= n
  \end{aligned}
\end{equation}

\section{Page Rank walk}
In the Page Rank random walk, at each step we toss a coin
with probability $u$ of giving heads. If the result is tail, we pick a
random neighbor and move to it. If the result is head, we jump to a
random node in the graph. That is, if $r_{ij}$ is the probability of
jumping from a node $i$ to a node $j$, we have
\begin{equation}
  r_{ij} = (1-u)\frac{a_{ij}}{d_i} + u \frac{1}{n} = (1-u)w_{ij} + \frac{u}{n}
\end{equation}
The master equation for the page rank walk follows the general schema
of the standard walk, and is given by
\begin{equation}
  \begin{aligned}
    p_i(t+1) &= \sum_j p_j(t) \omega_{ji} \\
       &= \sum_j (1-u) w_{ji} p_j(t) + \frac{u}{n} \sum_j p_j(t) \\
       &= \sum_j (1-u) w_{ji} p_j(t) + \frac{u}{n}
  \end{aligned}
\end{equation}
Defining the vector 
\begin{equation}
  \ket{p}(t) \dfeq [p_1(t),\ldots,p_n(t)]'
\end{equation}
That is
\begin{equation}
  \label{master}
  \ket{p}(t+1) = (1-u)\mtx{W}'\ket{p}(t) + \frac{u}{n} \ket{1_n}
\end{equation}
The stationary point of this iteration (occupancy probability at
steady state) is given by
\begin{equation}
  \label{looksie}
  \ket{p^*} = (1-u)\mtx{W}'\ket{p^*} + \frac{u}{n}\ket{1_n}
\end{equation}
Since $1\ge\lambda_i>-1$ and $u>0$, the matrix $\mtx{I}-(1-u)\mtx{W}$
has all the eigenvalues strictly positive and therefore it is non
singular. Eq. (\ref{looksie}) can then be solved:
\begin{equation}
  \ket{p^*} = \bigl[\mtx{I} - (1-u)\mtx{W}']^{-1}\frac{u}{n}\ket{1_n}
\end{equation}
Equation (\ref{master}) can be rewritten as 
\begin{equation}
  \label{difference}
  \begin{aligned}
      \ket{p}(t+1) -\ket{p}(t) &= -\bigl[\mtx{I} - (1-u)\mtx{W}']\ket{p}(t) + \frac{u}{n} \ket{1_n} \\
      &\dfeq -\mtx{Q}\ket{p}(t) + \frac{u}{n} \ket{1_n}
    \end{aligned} 
\end{equation}
with 
\begin{equation}
  \mtx{Q} = \bigl[\mtx{I} - (1-u)\mtx{W}']
\end{equation}
and the equilibrium can be rewritten as
\begin{equation}
  \ket{p}^* = \mtx{Q}^{-1}\frac{u}{n}\ket{1_n}
\end{equation}

\section{Dynamics of the Page Rank Walk}
The matrix $\mtx{Q} = \bigl[\mtx{I} - (1-u)\mtx{W}']$ has the same
eigenvectors as $\mtx{W}'$ and eigenvalues
\begin{equation}
  \mu_i = 1 - (1-u)\lambda_i
\end{equation}
with 
\begin{equation}
  u = \mu_1 < \mu_2 \le \dots \le \mu_n = 1 - (1-u)\lambda_n < 2 - u
\end{equation}
If $u=0$ the matrix $\mtx{Q}$ is singular. We shall not
consider this case, which reduces to the standard random
walk. If $u>0$, $\mtx{Q}$ can be decomposed as
\begin{equation}
  \label{bar}
  \begin{aligned}
    \mtx{Q} &= \mtx{T}\mtx{L}\mtx{T}' \\
    \mtx{L} &= \mbox{diag}(\mu_1,\ldots,\mu_n) \\
  \end{aligned}
\end{equation}
Note that we place the eigenvalues in an usual order: normally the
first eigenvalue is the largest while in out case $\mu_1$ is the
smaller. We do this to maintain a simpler correspondence between
$\mu_i$ and $\lambda_i$.

\bigskip

\remark 
Before we continue, we shall do a \dqt{sanity check} to verify that
indeed $p^*$ is a probability vector, that is, that $\sum_ip_i^*=1$
At equilibrium, we have
\begin{equation}
  \ket{p}^* = \mtx{Q}^{-1}\frac{u}{n}\ket{1_n} = \frac{u}{n}\mtx{T}\mtx{L}^{-1}\mtx{T}^{-1}\ket{1_n}
\end{equation}
From the shape of $\mtx{T}^{-1}$, we have
\begin{equation}
  \left.\mtx{T}^{-1}\ket{1_n}\right|_i = \sum_k \psi_{i,k}
\end{equation}
that is
\begin{equation}
  \mtx{T}^{-1}\ket{1_n} = \sum_h \sum_k \psi_{i,k} \ket{e_h}
\end{equation}
From this it follows
\begin{equation}
  \mtx{L}^{-1}\mtx{T}^{-1}\ket{1_n} = \sum_h \frac{1}{\mu_h}\sum_k \psi_{h,k} \ket{e_h}
\end{equation}
\begin{equation}
  \mtx{T}\mtx{L}^{-1}\mtx{T}^{-1}\ket{1_n} = \sum_h \frac{1}{\mu_h}\sum_k \psi_{h,k} \ket{\phi_h}
\end{equation}
That is
\begin{equation}
  \ket{p^*} = \frac{u}{n} \sum_h \frac{1}{\mu_h}\sum_k \psi_{h,k} \ket{\phi_h}
\end{equation}
and
\begin{equation}
  \sum_i p_i^* = \frac{u}{n} \sum_h \frac{1}{\mu_h}\sum_k \psi_{h,k} \sum_i \phi_{h,i}
\end{equation}
Because of lemma \ref{sum0}, the only term remaining is that with $h=1$ for which $\mu_1=u$,
therefore we have
\begin{equation}
  \sum_i p_i^* = \frac{1}{n} \sum_k \psi_{1,k} \sum_i \phi_{1,i} 
  \stackrel{\dagger}{=}  \frac{u}{n} \frac{1}{u} n = 1
\end{equation}
where the equality $\dagger$ comes from (\ref{sums}).
\eor

\separate

Let us now go back to the iteration (\ref{difference}). If we iterate it $\Delta$ times, assuming that the right-hand side 
is kept constant, we have
\begin{equation}
  \ket{p}(t+\Delta) -\ket{p}(t) = \Delta\bigl[-\mtx{Q}\ket{p}(t) + \frac{u}{n} \ket{1_n}\bigr]
\end{equation}
If we now consider $t$ a continuous variable, divide by $\Delta$ and
take the limit for $\Delta\rightarrow{0}$ we obtain the continuous
approximation of the walk, that is
\begin{equation}
  \label{foo}
  \frac{d}{d t}\ket{p} = -\mtx{Q}\ket{p} + \frac{u}{n} \ket{1_n}
\end{equation}
This equation can be interpreted as a diffusion equation under the
spatial operator $\mtx{Q}$. We are interested in studying the time
evolution of the probability as it approaches the steady state.

The solution of (\ref{foo}) is 
\begin{equation}
  \label{solution}
  \ket{p}(t) = \frac{u}{n}\mtx{Q}^{-1}\ket{1}_n + e^{-\mtx{Q}t}\ket{C}
\end{equation}
where $\ket{C}$ is an arbitrary vector that depends on the initial
conditions. Using the initial probabilities $\ket{p^0}$, we have
\begin{equation}
  \begin{aligned}
    \ket{p}(t) &= {\displaystyle \frac{u}{n} \bigl(\mtx{I}-e^{-\mtx{Q}t}\bigr)\mtx{Q}^{-1}\ket{1_n} + e^{-\mtx{Q}t}\ket{p^0}} \\
    &= {\displaystyle \frac{u}{n} \mtx{T}\bigl(\mtx{I}-e^{-\mtx{L}t}\bigr)\mtx{L}^{-1}\mtx{T}^{-1}\ket{1_n} + \mtx{T}e^{-\mtx{L}t}\mtx{T}^{-1}\ket{p^0} }
  \end{aligned}
\end{equation}
In order to analyze more in detail the dynamics of the walk, we move
to the eigenvector basis. Define $\ket{\zeta}=\mtx{T}^{-1}\ket{p}$, 
$\ket{\zeta^0}=\mtx{T}^{-1}\ket{p^0}$, and 
\begin{equation}
  \label{ttdef}
  \mtx{T}^{-1}\ket{1_n} = \ket{b} = [b_1, \ldots, b_n]'\ \ \ b_i=\sum_k \psi_{i,k}
\end{equation}
With these we obtain
\begin{equation}
  \frac{d}{d t}\ket{\zeta} = -\mtx{L}\ket{\zeta} + \frac{u}{n} \ket{b}
\end{equation}
with solution
\begin{equation}
  \ket{\zeta}(t) = \frac{u}{n}\bigl(\mtx{I} - e^{-\mtx{L}t}\bigr)\mtx{L}^{-1}\ket{b} + e^{-\mtx{L}t}\ket{\zeta^0}
\end{equation}
The matrix that determines the dynamics are in this case diagonal,
therefore each direction in the eigenvector space evolves
independently. Consider the direction of the first eigenvector,
corresponding to $\mu_1=u$:
\begin{equation}
  \begin{aligned}
    \zeta_1 &= \frac{u}{n}\bigl(1 - e^{-u t}\bigr)\left.\mtx{L}^{-1}\ket{b}\right|_1 + e^{-u t}\zeta_1^0 \\
    &= \frac{u}{n}\bigl(1 - e^{-u t}\bigr) \frac{1}{u} \sum_k \psi_{1,k} + e^{-u t}\left.\mtx{T}^{-1}\ket{p^0}\right|_1 \\
    &= 1 - e^{-u t} + e^{-u t}\left.\mtx{T}^{-1}\ket{p^0}\right|_1 
  \end{aligned}
\end{equation}
If the walk begins at a specific node $m$, then $p_k^0=\delta_{k,m}$ that is
\begin{equation}
  \left.\mtx{T}^{-1}\ket{p^0}\right|_1 = \sum_k \psi_{1,k}\delta_{k,m} = 1
\end{equation}
leading to $\zeta_1=1$. The first eigenvalue $\mu_1$ does not
contribute to the dynamics of the probability distribution, which is
determined uniquely by $\mu_2,\ldots,\mu_n$, the only eigenvalues that
actually depend on the structure of the graph.

%
%  To be completed
%

\section{Node Visitation}
In the previous section we studied the evolution of the occupancy
probability of the various nodes, en evolution that is controlled by the
values $\{\mu_i\}_{i\ge{2}}$.

We are now interested in a different aspect of the walk, namely how
fast we visit the nodes of the graph, and whether Page Rank allows us
to visit the nodes faster than a regular random walk. 

Let us consider a graph in which traveling a random edge leads us to a
node with an average of $\bar{q}$ neighbors%
\footnote{This is not the same as the average number of neighbors per
  node determined picking a node at random from the graph, a point
  that we shall consider later on.}%
, and let $m_t$ be the number of different nodes visited at time $t$.

If we consider a random node of the graph, the probability that the
node had already been visited is $m_t/n$, and the probability that the
node is a new discovery is $1-m_t/n$.

Suppose now we have just arrived at a node $v$ in the graph. There are
two ways in which this can happen: either we were at a neighbor $v'$
and from that we walked to one of the neighbors (this happened with
probability $1-u$), or we decide to jump to a random node of the graph
and with that jump we arrived at $v$ (this happened with probability
$u$). Once we are at $v$, we also have two alternatives as to what to
do: we can either move to a neighbor (with probability $1-u$) or jump
to a random node (with probability $u$).

Let us consider separately the two ways in which we may have arrived
at $v$ and the possible ways in which we can leave $v$.

\begin{description}
\item[Arrived from a neighbor (probability $1-u$):] the node $v$ has
  on average $\bar{q}$ neighbors and, among these, there is $v'$, the
  node we came from. At the node $v$ we toss a coin:
  \begin{description}
  \item[Move to a neighbor (probability $1-u$):] The average number of
    neighbors is $\bar{q}$; with probability $1/\bar{q}$ we shall go
    back to $v'$, and not visit any new node. With probability
    $1-1/\bar{q}$ we shall move to a node other than the neighbor we
    came from and in this case we have a probability $1-m_t/n$ of
    visiting a new neighbor. Considering the probability of choosing
    this option, we shall visit a new node with probability
    \begin{equation}
      (1-u)^2 \frac{\bar{q}-1}{\bar{q}}\Bigl(1 - \frac{m_t}{n} \Bigr)
    \end{equation}
  \item[Jump to a random node (probability $u$):] In this case, it
    doesn't really matter the fact that one of our neighbors has
    certainly be visited: we jump to a random node in the graph, so
    the probability of choosing this option and visiting a new node is
    \begin{equation}
      (1-u)u \Bigl(1 - \frac{m_t}{n} \Bigr)
    \end{equation}
  \end{description}
\item[Arrived from a jump (probability $u$):] In this case, it really
  doesn't matter how we decide to leave the node, either with another
  jump or through a neighbor: we are in a random area of the graph,
  our neighbors are random nodes. The probability of visiting a new
  node using this option is
    \begin{equation}
      u \Bigl(1 - \frac{m_t}{n} \Bigr)
    \end{equation}
\end{description}

Putting it all together, the probability that with a step of the walk
we visit a new neighbor is
\begin{equation}
  \label{putsch}
  \begin{aligned}
    &(1-u)\Bigl[(1-u)\frac{\bar{q}-1}{\bar{q}}\Bigl(1 - \frac{m_t}{n}\Bigr) + u\Bigl(1 - \frac{m_t}{n}\Bigr)\Bigr] + u \Bigl(1 - \frac{m_t}{n}\Bigr) \\
    &= \Bigl[(1-u)^2 \frac{\bar{q}-1}{\bar{q}} + (1-u)u + u\Bigr]\Bigl(1 - \frac{m_t}{n}\Bigr) \\
    &\dfeq \nu \Bigl(1 - \frac{m_t}{n}\Bigr)
  \end{aligned}
\end{equation}
where $\nu$ depends on $u$ and can be written, after some manipulation, as
\begin{equation}
  \label{nudef}
  \nu \dfeq - \frac{1}{\bar{q}} u^2 + \frac{2}{\bar{q}}u + \frac{\bar{q}-1}{\bar{q}}
\end{equation}

From (\ref{putsch}), we determine that the average increment of the number
of nodes visited in the random walk is
\begin{equation}
  m_{t+1}-m_{t} = \nu \Bigl(1 - \frac{m_t}{n}\Bigr)
\end{equation}
Taking infinitesimal time steps, we transform this into a differential
equation
\begin{equation}
  \frac{d}{d t} m(t) = \nu \Bigl(1 - \frac{m(t)}{n}\Bigr)
\end{equation}
which, assuming that we begin the walk from a single node ($m(0)=1$) gives
\begin{equation}
  m(t) = n(1-e^{-\nu t/n}) + e^{-\nu t/n}
\end{equation}
Note that this is a dynamic very similar to that of the probability
spread (the two are, of course, related, so this should not come as a
surprise). The visit is faster is $\nu$ is large and, deriving
(\ref{nudef}) with respect to $u$, we see that the maximum of $\nu$ is
for $u=1$. That is, the completely random visit is the fastest. This
is compatible with the results already obtained in \cite{riascos:12}.

Of course, in practical situations, this doesn't mean that the random
jumps on a social network are the best search strategy: one must
consider that nearby site often have related information, so a
strategy that stays around as long as there are resources and then
jumps to a random node (that is, something akin to ARS, but with
random jumps in lieu of L\'evy flights) might be optimal, This,
however, goes beyond the scope of this note.

\section{Transit Time}
We are now interested in determining the \dqt{speed} at which the walk
moves, that is, the average time that it takes, if we start from a
node $i$, to reach a node $j$. We begin by rewriting the master
equation for continuous time. Note that the values $p_{ij}(t)$ are
discrete probabilities on the nodes but probability densities if we
consider time. That is, if we start from node $i$, the probability of
arriving at node $j$ in the time interval $[t,t+dt]$ is
$p_{ij}(t)dt$. With this in mind, we can write the master equation as
\begin{equation}
  p_{ij}(t) = \delta_{ij}\delta(0) + \int_0^t p_{jj}(t-\tau)F_{ij}(\tau)d\tau
\end{equation}
where $\delta_{ij}$ is the Kronecker delta, $\delta(0)$ the Dirac's,
and $F_{ij}(t)$ is the probability density of first passage,
that is $F_{ij}(t)dt$ is the probability that, starting from node $i$,
we go for the first time through node $j$ in the interval
$[t,t+dt]$. Taking the Laplace transform of this relation, we have
\begin{equation}
  \tilde{p}_{ij}(s) = \delta_{ij} + \tilde{p}_{jj}(s) \tilde{F}_{ij}(s)
\end{equation}
that is
\begin{equation}
  \label{foobar}
  \tilde{F}_{ij}(s) = \frac{\tilde{p}_{ij}(s)-\delta_{ij}}{\tilde{p}_{jj}(s)}
\end{equation}
\begin{figure*}[htbp]
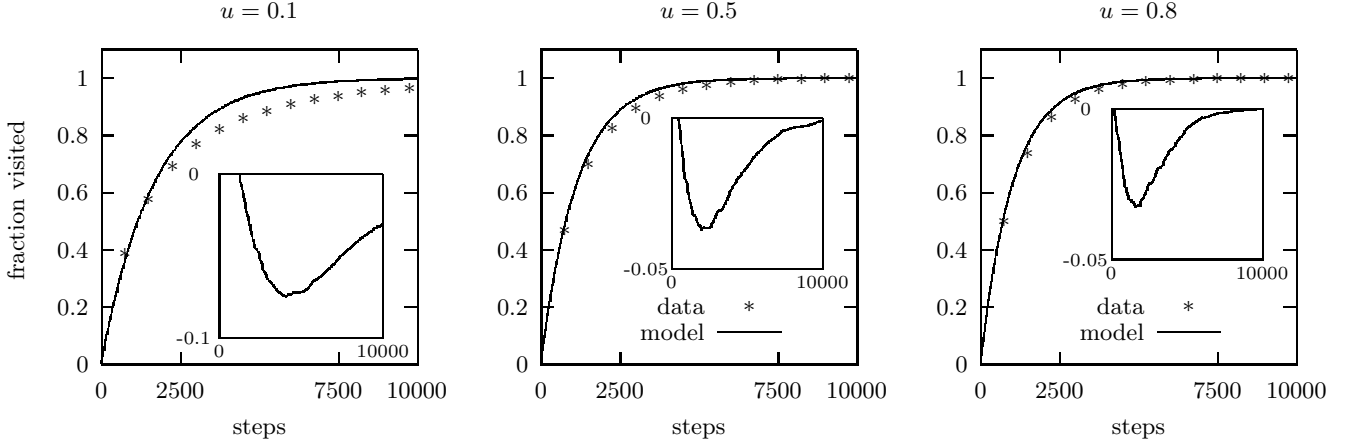

% GNUPLOT: LaTeX picture
\setlength{\unitlength}{0.240900pt}
\ifx\plotpoint\undefined\newsavebox{\plotpoint}\fi
\sbox{\plotpoint}{\rule[-0.200pt]{0.400pt}{0.400pt}}%

  \caption{The number of new nodes visited during the simulation
    (indicated as \dqt{*}) and the predicted value (continuous curve)
    for $u=0.1,0.5,0.8$. The insets show the error (simulated minus
    predicted). The variance of the data is very small, and was not
    indicated not to clutter the figures.}
  \label{occupancy}
\end{figure*}
This quantity is related to the expected time that we are looking
for. The expected time to go from $i$ to $j$ is the average of the
first passage, that is
\begin{equation}
  \langle T_{ij} \rangle = \int_0^{\infty}\!\!\!\!t F_{ij}(t)\,\,dt = - \tilde{F}_{ij}^\prime (0)
\end{equation}
In order to compute the derivative, we follow a procedure similar to
that in \cite{riascos:12}, translating it to continuous
time. We have
\begin{equation}
    \begin{aligned}
    \tilde{p}_{ij}(s) &= \int e^{-st} p_{ij}(t) dt \\
        &= p_j^* \int e^{-st} + \int e^{-st} \bigl[p_{ij}(t) - p_j^*\bigr] dt \\
        &\dfeq p_j^* \int e^{-st} + \int e^{-st} \hat{p}_{ij}(t) dt
    \end{aligned}   
\end{equation}
where we have defined $\hat{p}_{ij}(t)\dfeq{p_{ij}(t)}-p_j^*$. Computing
the first integral and expressing the exponential in the second as a
series, we have
\begin{equation}
    \begin{aligned}
      \tilde{p}_{ij}(s) &=  \frac{p_j^*}{s} + \sum_{n=0}^\infty \frac{(-1)^n}{n!} s^n \int t^n \hat{p}_{ij}(t) dt \\
  &\dfeq \frac{p_j^*}{s} + \sum_{n=0}^\infty \frac{(-1)^n}{n!} s^n R_{ij}^n \\
  &\dfeq \frac{p_j^*}{s} Q_{ij}(s)
  \end{aligned}
\end{equation}
where we have defined the moments
\begin{equation}
  R_{ij}^n = \int_0^\infty\!\!\!t^n \hat{p}_{ij}(t) dt
\end{equation}
and the shortcut
\begin{equation}
  Q_{ij}(s) = \sum_{n=0}^\infty \frac{(-1)^n}{n!} s^n R_{ij}^n
\end{equation}
Note that $Q_{ij}(0)=R_{ij}^{(0)}$. We can now rewrite (\ref{foobar}) as
\begin{equation}
    \begin{aligned}
      \tilde{F}_{ij}(s) &= \frac{1}{\displaystyle \frac{p_j^*}{s} + Q_{jj}(s)} \Bigl[ \frac{p_j^*}{s} + Q_{ij}(s) - \delta_{ij} \Bigr] \\
     &= \frac{p_j^* + s(Q_{ij}(s)-\delta_{ij})}{p+j^*+sQ_{jj}(s)}
     \end{aligned}
\end{equation}
It derivative is
\begin{equation}
    \begin{aligned}
  \tilde{F}_{ij}^\prime(s) &= \frac{Q_{ij}(s)-\delta_{ij}+sQ_{ij}^\prime(s)}{p_j^*+sQ_{jj}(s)} \\
    &-
  \frac{(Q_{jj}(s)+sQ_{jj}^\prime(s))\bigl[p_j^*+s(Q_{ij}(s)-\delta_{ij})\bigr]}{(p_j^*+sQ_{jj}(s))^2}
  \end{aligned} 
\end{equation}
Computing it in $s=0$ we obtain
\begin{equation}
  \langle T_{ij} \rangle = \frac{1}{p_j^*}\bigl[ R_{jj}^{(0)} - R_{ij}^{(0)} + \delta_{ij} \bigr]
\end{equation}
Taking the average over all pair of nodes, we have
\begin{equation}
  \label{tave}
  \langle T \rangle = \sum_{i \ne j} \langle T_{ij} \rangle p_j^* = \sum_i R_{ii}^{(0)}
\end{equation}

In order to compute the values $R_{ii}^{(0)}$ we need to express
$p_{ij}(t)$. From (\ref{solution}) we have
\begin{equation}
  p_{ij}(t) = \frac{u}{N} \bra{e_j} (\mtx{I}-e^{-\mtx{Q}t})\mtx{Q}^{-1}\ket{1_n} + \bra{e_j} e^{-\mtx{Q}t} \ket{e_i}
\end{equation}
that is
\begin{equation}
  \begin{aligned}
    \hat{p}_{ij}(t) &= \bra{e_j} e^{-\mtx{Q}t} \ket{e_i} - \frac{u}{N} \bra{e_j} e^{-\mtx{Q}t}\mtx{Q}^{-1}\ket{1_n} \\
    & = \bra{e_j} \mtx{T} e^{-\mtx{L}t} \mtx{T}^{-1} \ket{e_i} - 
    \frac{u}{N} \bra{e_j} \mtx{T} e^{-\mtx{L}t}\mtx{L}^{-1}\mtx{T}^{-1}\ket{1_n} 
  \end{aligned}
\end{equation}
Using the definitions (\ref{tdef}), (\ref{t1def}), and (\ref{ttdef}),
we can expand this equation as
\begin{equation}
  \hat{p}_{ij}(t) = \sum_k e^{-\mu_kt} \phi_{k,j} \psi_{k,i} - \frac{u}{N} \frac{e^{-\mu_kt}}{\mu_k} \phi_{k,j} b_k
\end{equation}
That is
\begin{equation}
    \begin{aligned}
  R_{ij}^{(0)} &= \int_0^\infty\!\!\!\hat{p}_{ij}(t) = 
  \sum_k \frac{1}{\mu_k} \Bigl[ \phi_{k,j}\psi_{k,i} - \frac{u}{N\mu_k}\phi_{k,j}b_k \Bigr] \\
  &= \sum_k \frac{1}{\mu_k} \psi_{k,i} \psi_{k,i} - \frac{u}{N\mu_k^2} \sum_k \phi_{k,i} b_k
  \end{aligned}
\end{equation}
Therefore, from (\ref{tave}) we have
\begin{equation}
    \begin{aligned}
  \langle T \rangle &= \sum_i \sum_k \frac{1}{\mu_k} \phi_{k,i}\psi_{k,i} - 
  \sum_i \sum_k \frac{u}{N\mu_k^2} \phi_{k,i} b_k \\
  &= 
  \sum_k \frac{1}{\mu_k} \sum_i  \phi_{k,i}\psi_{k,i} - 
  \sum_k \frac{u}{N\mu_k^2} b_k \sum_i \phi_{k,i}
  \end{aligned}
\end{equation}
The first summation over $i$ is equal to $1$ because of
(\ref{deltaeq}), while of the second only the first term remains, the
others being zero because of lemma \ref{sum0}. Considering that $\mu_1=u$ and 
applying (\ref{sums}), we have
\begin{equation}
  \begin{aligned}
    \langle T \rangle &= \sum_k \frac{1}{\mu_k} - 
    \frac{u}{N\mu_1^2} \sum_i \phi_{1,i} \sum_i \psi_{k,1} \\
    &= \sum_k \frac{1}{\mu_k} - \frac{1}{\mu_1} \\
    &= \sum_{k=2}^{N} \frac{1}{\mu_k} \\
    &= \sum_{k=2}^{N} \frac{1}{1-(1-u)\lambda_k} 
  \end{aligned}
\end{equation}
This result is directly comparable with that of \cite{riascos:12}. If
$u=0$, the result is the same as that of \cite{riascos:12} when
$\alpha\rightarrow\infty$. Note however that in \cite{riascos:12} the
dependence on the parameter $\alpha$ is hidden in the dependence of
$\lambda_k$ while, in our case, $\lambda_k$ is a constant depending
only on the structure of the graph, and the dependence on $u$ is
explicit.

Note that, just as in the case of the dynamics, the result does not
depend on $\mu_1$, but only on the eigenvalues $\mu_2,\ldots,\mu_n$
which describe the structure of the graph.
\section{Simulations}
We consider a Random graph created using the algorithm of Leskovec
\emph{et al.} \cite{leskovec:08} with 1,000 nodes. On the graph, we
execute 200 random walks starting at a randomly selected node and
terminating when all the nodes have been visited. Using these walks we
estimate the number of new nodes visited as a function of the number
of steps, as well as the average time to move from one node to another
(the first node of the estimation is, for each walk always the origin
of the walk).

Figure~\ref{occupancy} shows the number of nodes visited and the
predicted value for $u=0.1,0.5,0.8$.
The variance of teh data was very small, and it is not shown to avoid
clutter in the figure. The model slightly overestimates the number of
nodes, and the error is more pronounced for small $u$ (viz.\ for walks
closer to the standard random walk.

Figure~\ref{transit} shows the average transit time between nodes as a
function of $u$.
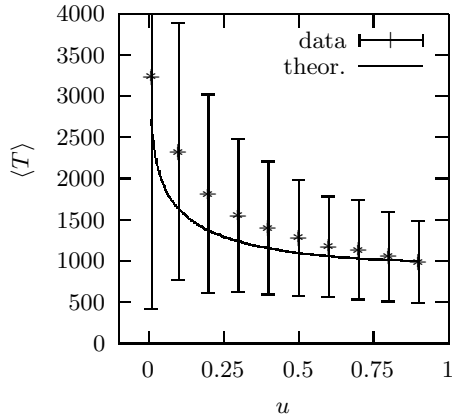
\begin{figure}
  \begin{center}
% GNUPLOT: LaTeX picture
\setlength{\unitlength}{0.240900pt}
\ifx\plotpoint\undefined\newsavebox{\plotpoint}\fi
\sbox{\plotpoint}{\rule[-0.200pt]{0.400pt}{0.400pt}}%
\begin{picture}(750,750)(0,0)
\sbox{\plotpoint}{\rule[-0.200pt]{0.400pt}{0.400pt}}%
\put(171.0,161.0){\rule[-0.200pt]{4.818pt}{0.400pt}}
\put(151,161){\makebox(0,0)[r]{$0$}}
\put(669.0,161.0){\rule[-0.200pt]{4.818pt}{0.400pt}}
\put(171.0,226.0){\rule[-0.200pt]{4.818pt}{0.400pt}}
\put(151,226){\makebox(0,0)[r]{$500$}}
\put(669.0,226.0){\rule[-0.200pt]{4.818pt}{0.400pt}}
\put(171.0,291.0){\rule[-0.200pt]{4.818pt}{0.400pt}}
\put(151,291){\makebox(0,0)[r]{$1000$}}
\put(669.0,291.0){\rule[-0.200pt]{4.818pt}{0.400pt}}
\put(171.0,355.0){\rule[-0.200pt]{4.818pt}{0.400pt}}
\put(151,355){\makebox(0,0)[r]{$1500$}}
\put(669.0,355.0){\rule[-0.200pt]{4.818pt}{0.400pt}}
\put(171.0,420.0){\rule[-0.200pt]{4.818pt}{0.400pt}}
\put(151,420){\makebox(0,0)[r]{$2000$}}
\put(669.0,420.0){\rule[-0.200pt]{4.818pt}{0.400pt}}
\put(171.0,485.0){\rule[-0.200pt]{4.818pt}{0.400pt}}
\put(151,485){\makebox(0,0)[r]{$2500$}}
\put(669.0,485.0){\rule[-0.200pt]{4.818pt}{0.400pt}}
\put(171.0,550.0){\rule[-0.200pt]{4.818pt}{0.400pt}}
\put(151,550){\makebox(0,0)[r]{$3000$}}
\put(669.0,550.0){\rule[-0.200pt]{4.818pt}{0.400pt}}
\put(171.0,614.0){\rule[-0.200pt]{4.818pt}{0.400pt}}
\put(151,614){\makebox(0,0)[r]{$3500$}}
\put(669.0,614.0){\rule[-0.200pt]{4.818pt}{0.400pt}}
\put(171.0,679.0){\rule[-0.200pt]{4.818pt}{0.400pt}}
\put(151,679){\makebox(0,0)[r]{$4000$}}
\put(669.0,679.0){\rule[-0.200pt]{4.818pt}{0.400pt}}
\put(218.0,161.0){\rule[-0.200pt]{0.400pt}{4.818pt}}
\put(218,120){\makebox(0,0){$0$}}
\put(218.0,659.0){\rule[-0.200pt]{0.400pt}{4.818pt}}
\put(336.0,161.0){\rule[-0.200pt]{0.400pt}{4.818pt}}
\put(336,120){\makebox(0,0){$0.25$}}
\put(336.0,659.0){\rule[-0.200pt]{0.400pt}{4.818pt}}
\put(454.0,161.0){\rule[-0.200pt]{0.400pt}{4.818pt}}
\put(454,120){\makebox(0,0){$0.5$}}
\put(454.0,659.0){\rule[-0.200pt]{0.400pt}{4.818pt}}
\put(571.0,161.0){\rule[-0.200pt]{0.400pt}{4.818pt}}
\put(571,120){\makebox(0,0){$0.75$}}
\put(571.0,659.0){\rule[-0.200pt]{0.400pt}{4.818pt}}
\put(689.0,161.0){\rule[-0.200pt]{0.400pt}{4.818pt}}
\put(689,120){\makebox(0,0){$1$}}
\put(689.0,659.0){\rule[-0.200pt]{0.400pt}{4.818pt}}
\put(171.0,161.0){\rule[-0.200pt]{0.400pt}{124.786pt}}
\put(171.0,161.0){\rule[-0.200pt]{124.786pt}{0.400pt}}
\put(689.0,161.0){\rule[-0.200pt]{0.400pt}{124.786pt}}
\put(171.0,679.0){\rule[-0.200pt]{124.786pt}{0.400pt}}
\put(30,420){\makebox(0,0){\begin{rotate}{90}$\langle{T}\rangle$\end{rotate}}}
\put(430,59){\makebox(0,0){$u$}}
\put(529,638){\makebox(0,0)[r]{data}}
\put(549.0,638.0){\rule[-0.200pt]{24.090pt}{0.400pt}}
\put(549.0,628.0){\rule[-0.200pt]{0.400pt}{4.818pt}}
\put(649.0,628.0){\rule[-0.200pt]{0.400pt}{4.818pt}}
\put(223.0,215.0){\rule[-0.200pt]{0.400pt}{111.778pt}}
\put(213.0,215.0){\rule[-0.200pt]{4.818pt}{0.400pt}}
\put(265.0,261.0){\rule[-0.200pt]{0.400pt}{97.083pt}}
\put(255.0,261.0){\rule[-0.200pt]{4.818pt}{0.400pt}}
\put(255.0,664.0){\rule[-0.200pt]{4.818pt}{0.400pt}}
\put(312.0,240.0){\rule[-0.200pt]{0.400pt}{75.161pt}}
\put(302.0,240.0){\rule[-0.200pt]{4.818pt}{0.400pt}}
\put(302.0,552.0){\rule[-0.200pt]{4.818pt}{0.400pt}}
\put(359.0,242.0){\rule[-0.200pt]{0.400pt}{57.816pt}}
\put(349.0,242.0){\rule[-0.200pt]{4.818pt}{0.400pt}}
\put(349.0,482.0){\rule[-0.200pt]{4.818pt}{0.400pt}}
\put(406.0,238.0){\rule[-0.200pt]{0.400pt}{50.348pt}}
\put(396.0,238.0){\rule[-0.200pt]{4.818pt}{0.400pt}}
\put(396.0,447.0){\rule[-0.200pt]{4.818pt}{0.400pt}}
\put(454.0,236.0){\rule[-0.200pt]{0.400pt}{43.844pt}}
\put(444.0,236.0){\rule[-0.200pt]{4.818pt}{0.400pt}}
\put(444.0,418.0){\rule[-0.200pt]{4.818pt}{0.400pt}}
\put(501.0,234.0){\rule[-0.200pt]{0.400pt}{38.062pt}}
\put(491.0,234.0){\rule[-0.200pt]{4.818pt}{0.400pt}}
\put(491.0,392.0){\rule[-0.200pt]{4.818pt}{0.400pt}}
\put(548.0,230.0){\rule[-0.200pt]{0.400pt}{37.580pt}}
\put(538.0,230.0){\rule[-0.200pt]{4.818pt}{0.400pt}}
\put(538.0,386.0){\rule[-0.200pt]{4.818pt}{0.400pt}}
\put(595.0,227.0){\rule[-0.200pt]{0.400pt}{33.967pt}}
\put(585.0,227.0){\rule[-0.200pt]{4.818pt}{0.400pt}}
\put(585.0,368.0){\rule[-0.200pt]{4.818pt}{0.400pt}}
\put(642.0,224.0){\rule[-0.200pt]{0.400pt}{31.076pt}}
\put(632.0,224.0){\rule[-0.200pt]{4.818pt}{0.400pt}}
\put(223,580){\makebox(0,0){$+$}}
\put(265,462){\makebox(0,0){$+$}}
\put(312,396){\makebox(0,0){$+$}}
\put(359,362){\makebox(0,0){$+$}}
\put(406,342){\makebox(0,0){$+$}}
\put(454,327){\makebox(0,0){$+$}}
\put(501,313){\makebox(0,0){$+$}}
\put(548,308){\makebox(0,0){$+$}}
\put(595,298){\makebox(0,0){$+$}}
\put(642,289){\makebox(0,0){$+$}}
\put(599,638){\makebox(0,0){$+$}}
\put(632.0,353.0){\rule[-0.200pt]{4.818pt}{0.400pt}}
\put(223,580){\makebox(0,0){$\ast$}}
\put(265,462){\makebox(0,0){$\ast$}}
\put(312,396){\makebox(0,0){$\ast$}}
\put(359,362){\makebox(0,0){$\ast$}}
\put(406,342){\makebox(0,0){$\ast$}}
\put(454,327){\makebox(0,0){$\ast$}}
\put(501,313){\makebox(0,0){$\ast$}}
\put(548,308){\makebox(0,0){$\ast$}}
\put(595,298){\makebox(0,0){$\ast$}}
\put(642,289){\makebox(0,0){$\ast$}}
\put(529,597){\makebox(0,0)[r]{theor.}}
\put(549.0,597.0){\rule[-0.200pt]{24.090pt}{0.400pt}}
\put(223,513){\usebox{\plotpoint}}
\multiput(223.60,495.15)(0.468,-6.038){5}{\rule{0.113pt}{4.300pt}}
\multiput(222.17,504.08)(4.000,-33.075){2}{\rule{0.400pt}{2.150pt}}
\multiput(227.60,460.21)(0.468,-3.552){5}{\rule{0.113pt}{2.600pt}}
\multiput(226.17,465.60)(4.000,-19.604){2}{\rule{0.400pt}{1.300pt}}
\multiput(231.59,439.94)(0.477,-1.823){7}{\rule{0.115pt}{1.460pt}}
\multiput(230.17,442.97)(5.000,-13.970){2}{\rule{0.400pt}{0.730pt}}
\multiput(236.60,423.19)(0.468,-1.797){5}{\rule{0.113pt}{1.400pt}}
\multiput(235.17,426.09)(4.000,-10.094){2}{\rule{0.400pt}{0.700pt}}
\multiput(240.60,411.43)(0.468,-1.358){5}{\rule{0.113pt}{1.100pt}}
\multiput(239.17,413.72)(4.000,-7.717){2}{\rule{0.400pt}{0.550pt}}
\multiput(244.60,401.85)(0.468,-1.212){5}{\rule{0.113pt}{1.000pt}}
\multiput(243.17,403.92)(4.000,-6.924){2}{\rule{0.400pt}{0.500pt}}
\multiput(248.60,393.68)(0.468,-0.920){5}{\rule{0.113pt}{0.800pt}}
\multiput(247.17,395.34)(4.000,-5.340){2}{\rule{0.400pt}{0.400pt}}
\multiput(252.59,387.59)(0.477,-0.599){7}{\rule{0.115pt}{0.580pt}}
\multiput(251.17,388.80)(5.000,-4.796){2}{\rule{0.400pt}{0.290pt}}
\multiput(257.60,381.09)(0.468,-0.774){5}{\rule{0.113pt}{0.700pt}}
\multiput(256.17,382.55)(4.000,-4.547){2}{\rule{0.400pt}{0.350pt}}
\multiput(261.60,375.51)(0.468,-0.627){5}{\rule{0.113pt}{0.600pt}}
\multiput(260.17,376.75)(4.000,-3.755){2}{\rule{0.400pt}{0.300pt}}
\multiput(265.00,371.94)(0.481,-0.468){5}{\rule{0.500pt}{0.113pt}}
\multiput(265.00,372.17)(2.962,-4.000){2}{\rule{0.250pt}{0.400pt}}
\multiput(269.00,367.94)(0.627,-0.468){5}{\rule{0.600pt}{0.113pt}}
\multiput(269.00,368.17)(3.755,-4.000){2}{\rule{0.300pt}{0.400pt}}
\multiput(274.00,363.94)(0.481,-0.468){5}{\rule{0.500pt}{0.113pt}}
\multiput(274.00,364.17)(2.962,-4.000){2}{\rule{0.250pt}{0.400pt}}
\multiput(278.00,359.94)(0.481,-0.468){5}{\rule{0.500pt}{0.113pt}}
\multiput(278.00,360.17)(2.962,-4.000){2}{\rule{0.250pt}{0.400pt}}
\multiput(282.00,355.95)(0.685,-0.447){3}{\rule{0.633pt}{0.108pt}}
\multiput(282.00,356.17)(2.685,-3.000){2}{\rule{0.317pt}{0.400pt}}
\multiput(286.00,352.95)(0.909,-0.447){3}{\rule{0.767pt}{0.108pt}}
\multiput(286.00,353.17)(3.409,-3.000){2}{\rule{0.383pt}{0.400pt}}
\multiput(291.00,349.95)(0.685,-0.447){3}{\rule{0.633pt}{0.108pt}}
\multiput(291.00,350.17)(2.685,-3.000){2}{\rule{0.317pt}{0.400pt}}
\put(295,346.17){\rule{0.900pt}{0.400pt}}
\multiput(295.00,347.17)(2.132,-2.000){2}{\rule{0.450pt}{0.400pt}}
\put(299,344.17){\rule{0.900pt}{0.400pt}}
\multiput(299.00,345.17)(2.132,-2.000){2}{\rule{0.450pt}{0.400pt}}
\multiput(303.00,342.95)(0.685,-0.447){3}{\rule{0.633pt}{0.108pt}}
\multiput(303.00,343.17)(2.685,-3.000){2}{\rule{0.317pt}{0.400pt}}
\put(307,339.17){\rule{1.100pt}{0.400pt}}
\multiput(307.00,340.17)(2.717,-2.000){2}{\rule{0.550pt}{0.400pt}}
\put(312,337.17){\rule{0.900pt}{0.400pt}}
\multiput(312.00,338.17)(2.132,-2.000){2}{\rule{0.450pt}{0.400pt}}
\put(316,335.17){\rule{0.900pt}{0.400pt}}
\multiput(316.00,336.17)(2.132,-2.000){2}{\rule{0.450pt}{0.400pt}}
\put(320,333.17){\rule{0.900pt}{0.400pt}}
\multiput(320.00,334.17)(2.132,-2.000){2}{\rule{0.450pt}{0.400pt}}
\put(324,331.67){\rule{1.204pt}{0.400pt}}
\multiput(324.00,332.17)(2.500,-1.000){2}{\rule{0.602pt}{0.400pt}}
\put(329,330.17){\rule{0.900pt}{0.400pt}}
\multiput(329.00,331.17)(2.132,-2.000){2}{\rule{0.450pt}{0.400pt}}
\put(333,328.17){\rule{0.900pt}{0.400pt}}
\multiput(333.00,329.17)(2.132,-2.000){2}{\rule{0.450pt}{0.400pt}}
\put(337,326.67){\rule{0.964pt}{0.400pt}}
\multiput(337.00,327.17)(2.000,-1.000){2}{\rule{0.482pt}{0.400pt}}
\put(341,325.67){\rule{1.204pt}{0.400pt}}
\multiput(341.00,326.17)(2.500,-1.000){2}{\rule{0.602pt}{0.400pt}}
\put(346,324.17){\rule{0.900pt}{0.400pt}}
\multiput(346.00,325.17)(2.132,-2.000){2}{\rule{0.450pt}{0.400pt}}
\put(350,322.67){\rule{0.964pt}{0.400pt}}
\multiput(350.00,323.17)(2.000,-1.000){2}{\rule{0.482pt}{0.400pt}}
\put(354,321.67){\rule{0.964pt}{0.400pt}}
\multiput(354.00,322.17)(2.000,-1.000){2}{\rule{0.482pt}{0.400pt}}
\put(358,320.67){\rule{1.204pt}{0.400pt}}
\multiput(358.00,321.17)(2.500,-1.000){2}{\rule{0.602pt}{0.400pt}}
\put(363,319.17){\rule{0.900pt}{0.400pt}}
\multiput(363.00,320.17)(2.132,-2.000){2}{\rule{0.450pt}{0.400pt}}
\put(367,317.67){\rule{0.964pt}{0.400pt}}
\multiput(367.00,318.17)(2.000,-1.000){2}{\rule{0.482pt}{0.400pt}}
\put(371,316.67){\rule{0.964pt}{0.400pt}}
\multiput(371.00,317.17)(2.000,-1.000){2}{\rule{0.482pt}{0.400pt}}
\put(375,315.67){\rule{0.964pt}{0.400pt}}
\multiput(375.00,316.17)(2.000,-1.000){2}{\rule{0.482pt}{0.400pt}}
\put(379,314.67){\rule{1.204pt}{0.400pt}}
\multiput(379.00,315.17)(2.500,-1.000){2}{\rule{0.602pt}{0.400pt}}
\put(384,313.67){\rule{0.964pt}{0.400pt}}
\multiput(384.00,314.17)(2.000,-1.000){2}{\rule{0.482pt}{0.400pt}}
\put(388,312.67){\rule{0.964pt}{0.400pt}}
\multiput(388.00,313.17)(2.000,-1.000){2}{\rule{0.482pt}{0.400pt}}
\put(392,311.67){\rule{0.964pt}{0.400pt}}
\multiput(392.00,312.17)(2.000,-1.000){2}{\rule{0.482pt}{0.400pt}}
\put(401,310.67){\rule{0.964pt}{0.400pt}}
\multiput(401.00,311.17)(2.000,-1.000){2}{\rule{0.482pt}{0.400pt}}
\put(405,309.67){\rule{0.964pt}{0.400pt}}
\multiput(405.00,310.17)(2.000,-1.000){2}{\rule{0.482pt}{0.400pt}}
\put(409,308.67){\rule{0.964pt}{0.400pt}}
\multiput(409.00,309.17)(2.000,-1.000){2}{\rule{0.482pt}{0.400pt}}
\put(396.0,312.0){\rule[-0.200pt]{1.204pt}{0.400pt}}
\put(418,307.67){\rule{0.964pt}{0.400pt}}
\multiput(418.00,308.17)(2.000,-1.000){2}{\rule{0.482pt}{0.400pt}}
\put(422,306.67){\rule{0.964pt}{0.400pt}}
\multiput(422.00,307.17)(2.000,-1.000){2}{\rule{0.482pt}{0.400pt}}
\put(426,305.67){\rule{0.964pt}{0.400pt}}
\multiput(426.00,306.17)(2.000,-1.000){2}{\rule{0.482pt}{0.400pt}}
\put(413.0,309.0){\rule[-0.200pt]{1.204pt}{0.400pt}}
\put(434,304.67){\rule{1.204pt}{0.400pt}}
\multiput(434.00,305.17)(2.500,-1.000){2}{\rule{0.602pt}{0.400pt}}
\put(430.0,306.0){\rule[-0.200pt]{0.964pt}{0.400pt}}
\put(443,303.67){\rule{0.964pt}{0.400pt}}
\multiput(443.00,304.17)(2.000,-1.000){2}{\rule{0.482pt}{0.400pt}}
\put(447,302.67){\rule{0.964pt}{0.400pt}}
\multiput(447.00,303.17)(2.000,-1.000){2}{\rule{0.482pt}{0.400pt}}
\put(439.0,305.0){\rule[-0.200pt]{0.964pt}{0.400pt}}
\put(456,301.67){\rule{0.964pt}{0.400pt}}
\multiput(456.00,302.17)(2.000,-1.000){2}{\rule{0.482pt}{0.400pt}}
\put(451.0,303.0){\rule[-0.200pt]{1.204pt}{0.400pt}}
\put(464,300.67){\rule{0.964pt}{0.400pt}}
\multiput(464.00,301.17)(2.000,-1.000){2}{\rule{0.482pt}{0.400pt}}
\put(460.0,302.0){\rule[-0.200pt]{0.964pt}{0.400pt}}
\put(473,299.67){\rule{0.964pt}{0.400pt}}
\multiput(473.00,300.17)(2.000,-1.000){2}{\rule{0.482pt}{0.400pt}}
\put(468.0,301.0){\rule[-0.200pt]{1.204pt}{0.400pt}}
\put(485,298.67){\rule{1.204pt}{0.400pt}}
\multiput(485.00,299.17)(2.500,-1.000){2}{\rule{0.602pt}{0.400pt}}
\put(477.0,300.0){\rule[-0.200pt]{1.927pt}{0.400pt}}
\put(494,297.67){\rule{0.964pt}{0.400pt}}
\multiput(494.00,298.17)(2.000,-1.000){2}{\rule{0.482pt}{0.400pt}}
\put(490.0,299.0){\rule[-0.200pt]{0.964pt}{0.400pt}}
\put(506,296.67){\rule{1.204pt}{0.400pt}}
\multiput(506.00,297.17)(2.500,-1.000){2}{\rule{0.602pt}{0.400pt}}
\put(498.0,298.0){\rule[-0.200pt]{1.927pt}{0.400pt}}
\put(515,295.67){\rule{0.964pt}{0.400pt}}
\multiput(515.00,296.17)(2.000,-1.000){2}{\rule{0.482pt}{0.400pt}}
\put(511.0,297.0){\rule[-0.200pt]{0.964pt}{0.400pt}}
\put(532,294.67){\rule{0.964pt}{0.400pt}}
\multiput(532.00,295.17)(2.000,-1.000){2}{\rule{0.482pt}{0.400pt}}
\put(519.0,296.0){\rule[-0.200pt]{3.132pt}{0.400pt}}
\put(545,293.67){\rule{0.964pt}{0.400pt}}
\multiput(545.00,294.17)(2.000,-1.000){2}{\rule{0.482pt}{0.400pt}}
\put(536.0,295.0){\rule[-0.200pt]{2.168pt}{0.400pt}}
\put(561,292.67){\rule{1.204pt}{0.400pt}}
\multiput(561.00,293.17)(2.500,-1.000){2}{\rule{0.602pt}{0.400pt}}
\put(549.0,294.0){\rule[-0.200pt]{2.891pt}{0.400pt}}
\put(583,291.67){\rule{0.964pt}{0.400pt}}
\multiput(583.00,292.17)(2.000,-1.000){2}{\rule{0.482pt}{0.400pt}}
\put(566.0,293.0){\rule[-0.200pt]{4.095pt}{0.400pt}}
\put(612,290.67){\rule{1.204pt}{0.400pt}}
\multiput(612.00,291.17)(2.500,-1.000){2}{\rule{0.602pt}{0.400pt}}
\put(587.0,292.0){\rule[-0.200pt]{6.022pt}{0.400pt}}
\put(617.0,291.0){\rule[-0.200pt]{6.022pt}{0.400pt}}
\put(171.0,161.0){\rule[-0.200pt]{0.400pt}{124.786pt}}
\put(171.0,161.0){\rule[-0.200pt]{124.786pt}{0.400pt}}
\put(689.0,161.0){\rule[-0.200pt]{0.400pt}{124.786pt}}
\put(171.0,679.0){\rule[-0.200pt]{124.786pt}{0.400pt}}
\end{picture}
  \end{center}
  \caption{The average time to transit between nodes as a function of
    $u$: simulation data vs.\ model predictions.}
  \label{transit}
\end{figure}
In this case, and coherently with the overestimation of
Figure~\ref{occupancy}, the model underestimates the actual time,
although it always remains within the standard deviation of the data,
and the error is especially pronounced for small $u$. The eigenvalues
of the $1000\times{1000}$ matrix used for the model were estimated
using an iterative method \cite{press:86}, and their accuracy was not
verified.

\section{Conclusions}
We have used a continuous approximation and spectral analysis to
determine the diffusion characteristics of Page Rank random walks. In
this, we have found some parallel with L\'evy Random walks, which also
give raise to superdiffusion \cite{riascos:14}. The usefulness of this
result is apparent especially in modeling searches: we now have a
plausible model of people's individual search behavior that can
provably generate walks whose statistical characteristics match those
of macroscopic search behaviors, such as ARS.

\appendix

\section{Eigenvalue $\lambda=-1$} 
\label{bipartite}
Let $\mtx{W}$ be the random walk transition matrix of a graph, which entails
\begin{equation}
  \label{positive}
  w_{ij} \ge 0
\end{equation}
as well as (\ref{coo}).
We want to prove the following:

\begin{theorem}
  $\mtx{W}$ has an eigenvalue $\lambda=-1$ iff the graph is bipartite%
  \footnote{Note that we determine the eigenvalue of $\mtx{W}$ while
    the equations that we use in the paper use $\mtx{W}'$. The two,
    however, have the same eigenvalues---but not the same
    eigenvectors, of course.}%
  %.
\end{theorem}

\begin{proof}
  Suppose first the graph is bipartite, with the two sets of nodes
  containing $n$ and $m$ nodes. We can label the nodes so that the
  nodes in the first set are $1,\ldots,n$ and those in the second are
  $n+1,\ldots,m$. This entails that the adjacency matrix $\mtx{A}$
  has structure
  \begin{equation}
    \mtx{A} = 
    \left[
      \begin{array}{cc}
        0 & \mtx{B} \\
        \mtx{B}' & 0
      \end{array}
      \right]
  \end{equation}
  and the walk matrix is
  \begin{equation}
    \label{wstruct}
    \mtx{W} = 
    \left[
      \begin{array}{cc}
        0 & \mtx{W}_1 \\
        \mtx{W}_2 & 0
      \end{array}
      \right]
  \end{equation}
  with $\mtx{W}_1\in{\mathbb{R}}^{n\times{m}}$ and
  $\mtx{W}_2\in{\mathbb{R}}^{m\times{n}}$. From the general properties
  of transition matrix, we have
  \begin{equation}
    \mtx{W}\ket{1_{n+m}} = \ket{1_{n+m}}
  \end{equation}
  From this and the structure of $\mtx{W}$, it follows
  \begin{equation}
    \label{weigen}
    \begin{aligned}
      \mtx{W}_1 \ket{1_m} &= \ket{1_n} \\
      \mtx{W}_2 \ket{1_n} &= \ket{1_m} 
    \end{aligned}
  \end{equation}
  Consider now the vector
  \begin{equation}
    \label{bdef}
    b = \left[
      \begin{array}{c}
        \ket{1_n} \\
        -\ket{1_m}
      \end{array}
      \right]
  \end{equation}
  For om the structure of $\mtx{W}$ and (\ref{weigen}), we have
  \begin{equation}
  \begin{aligned}
    \left[
      \begin{array}{cc}
        0 & \mtx{W}_1 \\
        \mtx{W}_2 & 0
      \end{array}
      \right]
    \,\left[
      \begin{array}{r}
        \ket{1_n} \\
        -\ket{1_m}
      \end{array}
      \right] 
    &=
    \left[
      \begin{array}{r}
        - \mtx{W}_1 \ket{1_m} \\
        \mtx{W}_2 \ket{1_n}
      \end{array}
      \right] \\
    &= 
    \left[
      \begin{array}{r}
        -\ket{1_n} \\
        \ket{1_m}
      \end{array}
      \right]
    = -b
    \end{aligned}
  \end{equation}
  Therefore $b$ is an eigenvector with eigenvalue $-1$.

  \separate

  Suppose now that $\mtx{W}$ has en eigenvalue $-1$. Let $b$ its
  eigenvector, and write
  \begin{equation}
    \ket{c} = \mtx{W}\,\ket{b}
  \end{equation}
  Clearly $c_i=-b_i$. 

  We first show that all the components of $b$ have the same absolute
  value. Suppose, by contradiction
  \begin{equation}
    \label{absval}
    b_h = \min_i\{b_i\} < \max_i \{b_i\} = b_k
  \end{equation}
  Then
  \begin{equation}
    |c_k| = \left| \sum_j w_{kj} b_j \right| \le 
    \sum_j w_{kj} |b_j| \stackrel{(\ddagger)}{<} |b_k| \sum_j w_{kj} = |b_k|
  \end{equation}
  where the inequality $(\ddagger)$ is strict because of (\ref{absval}) and
  the last equality is due to (\ref{coo}). This contradicts
  $c_i=-b_i$.

  Since all the elements have the same absolute value, we can consider
  a vector $b$ composed only of unitary elements, that is
  $b\in\{-1,1\}^{n+m}$, and we can shift rows and columns of $\mtx{W}$
  so that all the 1s are in the first positions, that is, $b$ is as in
  (\ref{bdef}). Consider now the element $c_i$ with $i\le{n}$. Since
  $b$ is an eigenvector for the eigenvalue $-1$, we must have:
  \begin{equation}
    c_i = \sum_j w_{ij} b_j = \sum_{j=1}^n w_{ij} - \sum_{j=n+1}^{n+m} w_{ij} = -b_i = -1
  \end{equation}
  Since $w_{ij}\ge{0}$ and (\ref{sums}) holds, the only way to obtain
  -1 is to have $w_{ij}=0$ for all $j=1,\ldots,n$, so that
  \begin{equation}
    \sum_{j=n+1}^{n+m} w_{ij} = \sum_{j=1}^{n+m} w_{ij} = 1
  \end{equation}
  The same argument applies to all $i=1,\ldots,n$. A similar
  argument for $i=n+1,\ldots,n+m$ shows that in that case we must have 
  $w_{ij}=0$ for $j=n+1,\ldots,n+m$.

  Putting the two together we have the condition
  \begin{equation}
    w_{ij}=0\ \ \ \mbox{if}\ \ (i\le n \mbox{ and } j\le m) \mbox{ or } 
    (i>n \mbox{ and } j>n)
  \end{equation}
  This condition gives us the structure (\ref{wstruct}) for
  $\mtx{W}$, which implies that the graph is bipartite.
\end{proof}


\begin{thebibliography}{10}

\bibitem{ali:18}
Imtiaz Ali, Ajay Gopinathan, et~al.
\newblock Levy walks in non-euclidean spaces.
\newblock {\em Bulletin of the American Physical Society}, 63, 2018.

\bibitem{barbour:90}
Andrew Barbour and Denis Mollison.
\newblock Epidemics and random graphs.
\newblock In {\em Stochastic processes in epidemic theory}, pages 86--9.
  Springer, 1990.

\bibitem{biyikoglu:07}
T{\"u}rker Biyikoglu, Josef Leydold, and Peter~F Stadler.
\newblock {\em Laplacian eigenvectors of graphs: Perron-Frobenius and
  Faber-Krahn type theorems}.
\newblock Springer, 2007.

\bibitem{borgers:20}
Christoph Bo\"orgers and Claude Greengard.
\newblock On the mean square displacement in levy walks.
\newblock {\em SIAM Journal on Applied Mathematics}, 80(3):1175--96, 2020.

\bibitem{boyer:06}
Denis Boyer, Gabriel Ramos-Fern{\'a}ndez, Octavio Miramontes, Jos{\'e}~L
  Mateos, Germinal Cocho, Hern{\'a}n Larralde, Humberto Ramos, and Fernando
  Rojas.
\newblock Scale-free foraging by primates emerges from their interaction with a
  complex environment.
\newblock {\em Proceedings of the Royal Society B: Biological Sciences},
  273(1595):1743--50, 2006.

\bibitem{brin:98}
Sergey Brin and Lawrence Page.
\newblock The anatomy ofa large-scale hypertextual web search engine.
\newblock In {\em Proceedings of the Seventh International World Wide Web
  Conference}, 1998.

\bibitem{britton:07}
Tom Britton, Svante Janson, and Anders Martin-L{\"o}f.
\newblock Graphs with specified degree distributions, simple epidemics, and
  local vaccination strategies.
\newblock {\em Advances in Applied Probability}, 39(4):922--48, 2007.

\bibitem{brockmann:06}
Dirk Brockmann, Lars Hufnagel, and Theo Geisel.
\newblock The scaling laws of human travel.
\newblock {\em Nature}, 439(7075):462--5, 2006.

\bibitem{gonzalez:08}
Marta~C Gonzalez, Cesar~A Hidalgo, and Albert-Laszlo Barabasi.
\newblock Understanding individual human mobility patterns.
\newblock {\em nature}, 453(7196):779--782, 2008.

\bibitem{iribarren:09}
Jos{\'e}~Luis Iribarren and Esteban Moro.
\newblock Impact of human activity patterns on the dynamics of information
  diffusion.
\newblock {\em Physical review letters}, 103(3):038702, 2009.

\bibitem{leskovec:08}
Jure Leskovec, Lars Backstrom, Ravi Kumar, and Andrew Tomkins.
\newblock Microscopic evolution of social networks.
\newblock In {\em Proceedings of the 14th ACM SIGKDD international conference
  on Knowledge discovery and data mining}, pages 462--470. ACM, 2008.

\bibitem{lovasz:96}
L\'aszl\'o Lov\'asz.
\newblock Random walks on graphs: a survey.
\newblock In D.~Milos, V.~T. Sos, and T.~Szony, editors, {\em Combinatorics,
  Paul Erd\"os is Eighty}, pages 353--98. Budapest,J\'anos Bolyai MAthematical
  Society, 1996.

\bibitem{page:99}
Lawrence Page, Sergey Brin, Rajeev Motwani, and Terry Winograd.
\newblock The {PageRank} citation ranking: Bringing order to the web.
\newblock Technical report, Stanford InfoLab, 1999.

\bibitem{press:86}
William~H. Press, Brian~P. Flannery, Saul~A. Teulolsky, and William~T.
  Vetterling.
\newblock {\em Numerical Recipes, The Art of Scientific Computing}.
\newblock Cambridge University Press, 1986.

\bibitem{ramos:04}
Gabriel Ramos-Fern{\'a}ndez, Jos{\'e}~L Mateos, Octavio Miramontes, Germinal
  Cocho, Hern{\'a}n Larralde, and Barbara Ayala-Orozco.
\newblock L{\'e}vy walk patterns in the foraging movements of spider monkeys
  (ateles geoffroyi).
\newblock {\em Behavioral ecology and Sociobiology}, 55(3):223--30, 2004.

\bibitem{riascos:12}
A.~P. Riascos and Jos\'e~L. Mateos.
\newblock Long-range navigation on complex networks using {L\'evy} random
  walks.
\newblock {\em Physical Review E}, 86(056110), 2012.

\bibitem{riascos:14}
A.~P. Riascos and Jos\'e~L. Mateos.
\newblock Fractional dynamics on networks: Emergence of anomalous diffusion and
  {L\'evy} flights.
\newblock {\em Physical Review E}, 90(032809), 2014.

\bibitem{santini:20a}
Simone Santini.
\newblock A random walk on area restricted search.
\newblock {\em ArXiv}, 2006.14318, 2020.

\bibitem{vespignani:12}
Alessandro Vespignani.
\newblock Modelling dynamical processes in complex socio-technical systems.
\newblock {\em Nature physics}, 8(1):32--9, 2012.

\end{thebibliography}
\end{document}